\documentclass[12pt]{article}
\usepackage[retainorgcmds]{IEEEtrantools}

\usepackage{cite}

\usepackage[cmex10]{amsmath}
\usepackage[tight,footnotesize]{subfigure}

\usepackage{enumerate}
\usepackage{amsmath}
\usepackage{amssymb}
\usepackage{amsthm}

\RequirePackage{mleftright} 
\mleftright                 

\newcommand{\dd}{\mathop{}\!\mathrm{d}}
\newcommand{\DTSPsin}[2]{\{ #1_{#2}\}_{#2 \in \Naturals}}
\newcommand{\DTSPdou}[2]{\{ #1_{#2}\}_{#2 \in \Integers}} 
\newcommand{\DTSPgen}[2]{\{ #1_{#2}\}}

\newcommand{\vect}[1]{\mathbf{#1}} 
\newcommand{\I}[1]{\operatorname{I}\{#1\}}

\newcommand{\bfx}{{\mathbf x}}

\newcommand{\Prs}[1]{\operatorname{\textnormal{Pr}}(#1)}

\newcommand{\Prv}[1]{\operatorname{\textnormal{Pr}}[#1]}

\newcommand{\eps}{\epsilon}
\newcommand{\bfZ}{{\mathbf Z}}

\newcommand{\Reals}{\mathbb R} 

\newcommand{\abs}[1]{\lvert#1\rvert} 
\newcommand{\card}[1]{\abs{#1}} 

\newcommand{\set}[1]{\mathcal{#1}} 
\newcommand{\eqdef}{\triangleq} 
\newcommand{\E}[1]{\operatorname{E}[#1]} 
\newcommand{\Naturals}{\mathbb{N}}
\newcommand{\Integers}{\mathbb{Z}}
\newcommand{\intinf}{\int_{-\infty}^{\infty}}

\renewcommand{\epsilon}{\varepsilon}

\newcommand{\mat}[1]{\mathsf{#1}}
\newcommand{\Normal}[2]{\mathcal{N}\left({#1},{#2}\right)} 
\newcommand{\trans}[1]{#1^{\textnormal{\textsf{\tiny T}}}} 

\newcommand{\weaktyp}{\set{T}_{n}^{\epsilon}}
\newcommand{\contyp}[1]{\set{G}_{n}^{\epsilon}(#1)}

\newtheorem{theorem}{Theorem}
\newtheorem{proposition}[theorem]{Proposition}
\newtheorem{lemma}[theorem]{Lemma}

\newtheorem{remark}[theorem]{Remark}

\author{Christoph Bunte and Amos Lapidoth\thanks{This work was
    presented in part at the Seventh Joint Workshop on Coding and
    Communications (JWCC) 2014 November 13-–15, 2014, Barcelona,
    Spain}}

\title{Maximum R\'enyi Entropy  Rate}

\begin{document}

\maketitle

\begin{abstract}
  Two maximization problems of R\'enyi entropy rate are investigated:
  the maximization over all stochastic processes whose marginals
  satisfy a linear constraint, and the Burg-like maximization 
  over all stochastic processes whose autocovariance function begins
  with some given values. The solutions are related to the solutions
  to the analogous maximization problems of Shannon entropy rate.
\end{abstract}


\section{Introduction}
Motivated by recent results providing an operational meaning to
R\'enyi entropy \cite{BunteLapidothTasmenia}, we study the
maximization of the R\'enyi entropy rate (or ``R\'enyi rate'') over
the class of stochastic processes $\DTSPdou{Z}{k}$
that satisfy
\begin{equation}
  \label{eq:AAconstSP}
  \Prv{Z_{k} \in \set{S}} = 1, \quad \E{r(Z_{k})} \leq \Gamma, \quad
  k \in \Integers,
\end{equation}
where $\set{S} \subseteq \Reals$ is some given support set, $r(\cdot)$
is some cost function, $\Gamma \in \Reals$ is some maximal-allowed
average cost, and $\Reals$ and $\Integers$ denote the reals and the
integers respectively.

If instead of R\'enyi rate we had maximized the Shannon rate, we could
have limited ourselves to memoryless processes, because the Shannon
entropy of a random vector is upper-bounded by the sum of the Shannon
entropies of its components, and this upper bound is tight when the
components are independent.\footnote{Throughout this paper ``Shannon entropy''
  refers to differential Shannon entropy.} But this bound does not hold for R\'enyi
entropy: the R\'enyi entropy of a vector with dependent components can
exceed the sum of the R\'enyi entropies of its
components. Consequently, the solution to the maximization of the
R\'enyi rate subject to~\eqref{eq:AAconstSP} is typically not
memoryless. This maximum and the structure of the stochastic processes
that approach it is the subject of this paper.

Another class of stochastic processes that we shall consider is
related to Burg's work on spectral estimation \cite{Burg},
\cite[Theorem~12.6.1]{cover2006elements}. It comprises all (one-sided)
stochastic processes $\DTSPsin{X}{i}$
that, for some given $\alpha_{0}, \ldots, \alpha_{p} \in \Reals$,
satisfy
\begin{IEEEeqnarray}{c}
  \E{X_i X_{i+k}} = \alpha_k, \quad \Bigl( i\in\Naturals, \; k\in \{0,
  \ldots, p\} \Bigr),
  \label{eq:Burg_cons}
\end{IEEEeqnarray}
where $\Naturals$ denotes the positive integers. While Burg studied the
maximum over this class of the Shannon rate, we will study the maximum
of the R\'enyi rate.

We emphasize that our focus here is on the maximization of R\'enyi
rate and not entropy. The latter is studied in \cite{Eilat2014},
\cite{SundaresanKumarForward}, \cite{LutwakYangZhang2004}, and
\cite{costa2003solutions}.

To describe our results we need some definitions.  The order-$\alpha$
R\'enyi entropy of a probability density function (PDF) $f$ is defined
as
\begin{equation}
  \label{eq:renyi}
  h_{\alpha}(f) = \frac{1}{1-\alpha} \log \intinf f^{\alpha}(x) \dd{x}, 
\end{equation}
where $\alpha$ can be any positive number other than one. The
integrand is nonnegative, so the integral
on the RHS of~\eqref{eq:renyi} always exists, possibly taking on the
value $+\infty$, in which case we define $h_{\alpha}(f)$ as $+\infty$ if
$0<\alpha<1$ and as $-\infty$ if $\alpha>1$. With this
convention the R\'enyi entropy always exists and 
\begin{align}
  \label{eq:h_ge_minus}
  h_{\alpha}(f) & > -\infty, \quad 0 < \alpha < 1, \\
  h_{\alpha}(f) & < +\infty, \quad \alpha > 1. \label{eq:h_le_plus}
\end{align}
When a random variable (RV) $X$ is of density $f_{X}$ we sometimes
write $h_{\alpha}(X)$ instead of $h_{\alpha}(f_{X})$. The R\'enyi
entropy of some multivariate densities are computed in \cite{ZografosNadarjah2005}.

If the support of $f$ is contained in $\set{S}$, then
\begin{equation}
  \label{eq:logSupp}
  h_{\alpha}(f) \leq \log \> \card{\set{S}}, \quad 
  \Bigl(\alpha > 0, \; \alpha \neq 1\Bigr),
\end{equation}
where $\card{\set{A}}$ denotes the Lebesgue measure of the set
$\set{A}$, and where we interpret $\log \> \card{\set{S}}$ as $+\infty$
when $\card{\set{S}}$ is infinite. (Throughout this paper we define
$\log \infty = \infty$ and $\log 0 = -\infty$.)

The R\'enyi entropy is closely related to the Shannon entropy:
\begin{equation}
  \label{eq:shannon}
  h(f) = -\intinf f(x) \log f(x) \dd{x}.
\end{equation}
(The integral on the RHS of~\eqref{eq:shannon} need not exist. If it
does not, then we say that $h(f)$ does not exist.) 
Depending on whether $\alpha$ is smaller or larger than one, the
R\'enyi entropy can be larger or smaller than the Shannon
entropy. Indeed, if $f$ is of Shannon entropy $h(f)$ (possibly
$+\infty$), then by \cite[Lemma~5.1 (iv)]{WangMadiman2014}:
\begin{IEEEeqnarray}{rCl"s}
  h_{\alpha}(f) & \leq & h(f), &  for $\alpha>1$;
  \label{eq:Ren_le_Shannon}
  \\
  h_{\alpha}(f) & \geq & h(f), & for $0<\alpha<1$.
  \label{eq:Ren_ge_Shannon}
\end{IEEEeqnarray}
Moreover, under some mild technical conditions \cite[Lemma~5.1
(ii)]{WangMadiman2014}:
\begin{equation}
  \label{eq:limRen}
  \lim_{\alpha\to 1} h_{\alpha}(f)=h(f). 
\end{equation}

The order-$\alpha$ R\'enyi rate $h_{\alpha}(\DTSPgen{X}{k})$ of a
stochastic process (SP) $\DTSPgen{X}{k}$ is defined as
\begin{equation}
  \label{eq:def_rate}
  h_{\alpha}(\DTSPgen{X}{k}) = \lim_{n \to \infty} \frac{1}{n} 
  h_{\alpha}\bigl(X_{1}^{n}\bigr)
\end{equation}
whenever the limit exists.\footnote{We say that the limit exists and
  is equal to $+\infty$ if for every $\mathsf{M} > 0$ there exists
  some $n_{0}$ such that for all $n > n_{0}$ the R\'enyi entropy
  $h_{\alpha}(X_{1}, \ldots, X_{n})$ exceeds $n \mathsf{M}$, possibly
  by being $+\infty$.}  Here $X_{i}^{j}$ denotes the tuple $(X_{i},
\ldots, X_{j})$. 

Notice that if each $X_{k}$ takes value in $\set{S}$, then $X_{1}^{n}$
takes value in $\set{S}^{n}$, and it then follows
from~\eqref{eq:logSupp} that $h_{\alpha}(X_{1}^{n}) \leq
\log \card{\set{S}}^{n}$ and thus
\begin{equation}
  \label{eq:SP_ub_S}
  h_{\alpha}(\DTSPgen{X}{k}) \leq \log \card{\set{S}}.
\end{equation}
Another upper bound on $h_{\alpha}(\{X_{k}\})$, one that is valid for
$\alpha > 1$, can be obtained by noting that when $\alpha > 1$ we can
use~\eqref{eq:Ren_le_Shannon} to obtain
\begin{align}
  h_{\alpha}(X_{1}^{n}) & \leq h(X_{1}^{n}) \label{eq:Orthod10}\\
  & \leq \sum_{i=1}^{n} h(X_{i}), \label{eq:up_by_sum_Shannon}
\end{align}
and thus, by~\eqref{eq:Orthod10},
\begin{equation}
  \label{eq:SP_up_Shannon}
  h_{\alpha}(\DTSPgen{X}{k}) \leq h(\DTSPgen{X}{k}), \quad \alpha > 1,
\end{equation}
whenever both $h_{\alpha}(\DTSPgen{X}{k})$ and the Shannon rate
$h(\DTSPgen{X}{k})$ exist.

The R\'enyi rate of finite-state Markov chains was computed by Rached,
Alajaji, and Campbell \cite{Rached_Alajaji_Campbell} with extensions
to countable state space in \cite{GolshaniPasha_Yair}.  The R\'enyi
rate of stationary Gaussian processes was found by Golshani and Pasha
in \cite{GolshaniPasha_Gauss}. Extensions are explored in
\cite{Khodabin_ADK}.

\section{Main Results}
We discuss the constraints~\eqref{eq:AAconstSP}
and~\eqref{eq:Burg_cons} separately. The proofs pertaining to the
former are in Section~\ref{sec:Proofs1} and to the latter in
Section~\ref{sec:proofburgrenyi}.

\subsection{Max R\'enyi Rate Subject to~\eqref{eq:AAconstSP}}

Let $h^{\star}(\Gamma)$ denote the supremum of $h(f_{X})$ over all
densities $f_{X}$ under which
\begin{equation}
  \label{eq:constraints_on_X}
  \Prs{X \in \set{S}} = 1 
  \quad \text{and} \quad
  \E{r(X)} \leq \Gamma.
\end{equation}
Here and throughout the supremum should be interpreted as $-\infty$
whenever the maximization is over an empty set. Thus, if no
distribution satisfies \eqref{eq:constraints_on_X}, then
$h^{\star}(\Gamma)$ is $-\infty$.

We shall assume that for some $\Gamma_{0} \in \Reals$ 
\begin{subequations}
  \label{block:amos}
  \begin{equation}
    \label{eq:h_Gamma0_ge}
    h^{\star}(\Gamma_{0}) > -\infty,
  \end{equation}
  and 
  \begin{equation}
    \label{eq:finite_for_geq}
    h^{\star}(\Gamma) < \infty \quad \text{for every $\Gamma \geq \Gamma_{0}$}.
  \end{equation}
\end{subequations}
Under this assumption the function $h^{\star}$ has the following properties:
\begin{proposition}
  \label{prop:HstarProperties}
  Let $\Gamma_{0}$ satisfy \eqref{block:amos}. Then over the interval
  $[\Gamma_{0},\infty)$ the function $h^{\star}(\cdot)$ is finite,
  nondecreasing, and concave. It is continuous over
  $(\Gamma_{0},\infty)$, and
  \begin{equation}
    \label{eq:Lim_h_star}
    \lim_{\Gamma \to \infty}
    h^{\star}(\Gamma) = \log \,\card{\set{S}}.
  \end{equation}
\end{proposition}
\begin{proof}
  Monotonicity is immediate from the definition because increasing
  $\Gamma$ enlarges the set of densities that
  satisfy~\eqref{eq:constraints_on_X}.  Concavity follows from the
  concavity of Shannon entropy, and continuity follows from
  concavity. It remains to establish~\eqref{eq:Lim_h_star}. To this end we first
  argue that for every $\Gamma$,
  \begin{equation}
    \label{eq:max_Shannon_uniform}
    h^{\star}(\Gamma) \leq \log \,\card{\set{S}}.
  \end{equation}
  When $\card{\set{S}}$ is infinite this is trivial, and when
  $\card{\set{S}}$ is finite this follows by noting that
  $h^{\star}(\Gamma)$ cannot exceed the maximum of the Shannon entropy
  in the absence of cost constraints, and the latter is achieved by a
  uniform distribution on $\set{S}$ and is equal to $\log
  \,\card{\set{S}}$. In view of~\eqref{eq:max_Shannon_uniform}, our
  claim~\eqref{eq:Lim_h_star} will follow once we establish that
  \begin{equation}
    \label{eq:lim_need_liminf}
    \varliminf_{\Gamma \to \infty}
    h^{\star}(\Gamma) \geq \log \,\card{\set{S}},
  \end{equation}
  which is what we set out to prove next. 

  We first note that for every $\Gamma \in \Reals$
  \begin{equation}
    \label{eq:Aamos10}
    h^{\star}(\Gamma) \geq \log \, \card{\{x\in \set{S} \colon r(x)\leq
      \Gamma\}} 
  \end{equation}
  because when the RHS is finite it can be achieve by a uniform
  distribution on the set $\{x\in\set{S} \colon r(x) \leq \Gamma \}$, a
  distribution under which~\eqref{eq:constraints_on_X} clearly holds,
  and when it is infinite, it can be approached by uniform
  distributions on ever-increasing compact subsets of this set. We
  next note that, by the Monotone Convergence Theorem (MCT),
  \begin{equation}
    \label{eq:Aamos20}
    \lim_{\Gamma\to\infty} \card{\{x\in\set{S} \colon r(x) \leq \Gamma\}} 
    = \card{\set{S}}.
  \end{equation}
  Combining~\eqref{eq:Aamos10} and~\eqref{eq:Aamos20}
  establishes~\eqref{eq:lim_need_liminf} and hence completes the proof
  of~\eqref{eq:Lim_h_star}.
\end{proof}

For $\alpha > 1$ we note that~\eqref{eq:def_rate},
\eqref{eq:up_by_sum_Shannon}, and the definition of
$h^{\star}(\Gamma)$ imply that for every SP $\DTSPgen{Z}{k}$
satisfying~\eqref{eq:AAconstSP}
\begin{equation}
  \label{eq:ConverseBiggerThan1}
  h_{\alpha}(\DTSPgen{Z}{k}) \leq h^{\star}(\Gamma), \quad \alpha > 1,
\end{equation}
and consequently,
\begin{equation}
  \label{eq:AAAConverseBiggerThan1}
  \sup h_{\alpha}(\DTSPgen{Z}{k}) \leq h^{\star}(\Gamma), \quad \alpha > 1,
\end{equation}
where the supremum is over all SPs
satisfying~\eqref{eq:AAconstSP}. Perhaps surprisingly, this bound is tight:
\begin{theorem}[Max R\'enyi Rate for $\alpha > 1$]
  \label{thm:AlphaBig}
  Suppose that $\alpha > 1$, and that $\Gamma > \Gamma_{0}$, where
  $\Gamma_{0}$ satisfies \eqref{block:amos}. Then for every
  $\tilde{\eps} > 0$ there exists a stationary SP $\DTSPgen{Z}{k}$
  satisfying~\eqref{eq:AAconstSP}
  whose R\'enyi rate is defined and exceeds $h^{\star}(\Gamma) -
  \tilde{\eps}$.
\end{theorem}

For $0 < \alpha < 1$ we can use~\eqref{eq:SP_ub_S} to obtain for the
same supremum
\begin{equation}
  \label{eq:AAAConverseSmallerThan1}
  \sup h_{\alpha}(\DTSPgen{Z}{k}) \leq \log \card{\set{S}}, \quad 0 < \alpha < 1.
\end{equation}
This seemingly crude bound is tight:
\begin{theorem}[Max R\'enyi Rate for $0 < \alpha < 1$]
  \label{thm:AlphaSmall}
  Suppose that $0 < \alpha < 1$ and that $\Gamma > \Gamma_{0}$, where
  $\Gamma_{0}$ satisfies \eqref{block:amos}.
  \begin{itemize} 
\item If $|\set{S}| = \infty$, then for every $\mathsf{M} \in \Reals$ there
  exists a stationary SP $\DTSPgen{Z}{k}$ satisfying~\eqref{eq:AAconstSP}
  whose R\'enyi rate is defined and exceeds $\mathsf{M}$.
  \item If $|\set{S}| < \infty$, then for every $\tilde{\eps} > 0$ there
      exists a stationary SP $\DTSPgen{Z}{k}$ satisfying~\eqref{eq:AAconstSP}
  whose R\'enyi rate is defined and exceeds $\log |\set{S}| -
  \tilde{\eps}$.
    \end{itemize}
\end{theorem}

\begin{remark} Theorems~\ref{thm:AlphaBig} and~\ref{thm:AlphaSmall}
  can be generalized in a
straightforward fashion to account for multiple constraints:
\begin{equation}
  \E{r_{i}(Z_k)} \leq \Gamma_{i}, \quad i=1,\ldots,m.
\end{equation}
However, for ease of presentation we focus on the case of a single
constraint.
\end{remark}

A special case of Theorems~\ref{thm:AlphaBig} and~\ref{thm:AlphaSmall}
is when the cost is quadratic, i.e., $r(x) = x^{2}$ and where there are no
restrictions on the support, i.e., $\set{S} = \Reals$. In this case we
can slightly strengthen the results of the above theorems: When we
consider the proofs of these theorems for this case, we see that the
proposed distributions are isotropic. We can thus establish that the
constructed SP is centered and uncorrelated:
\begin{proposition}[R\'enyi Rate under a Second-Moment Constraint] 
  \label{cor:SecondMoment}
  \mbox{}
  \begin{enumerate}
  \item For every $\alpha > 1$, every $\sigma > 0$, and every
    $\tilde{\eps} > 0$ there exists a centered stationary SP
    $\{Y_{k}\}$ whose R\'enyi rate exceeds $\frac{1}{2}\log (2 \pi e
    \sigma^{2}) - \tilde{\eps}$ and that satisfies
    \begin{equation}
      \label{eq:VarianceP}
      \E{Y_{k}Y_{k'}} = \sigma^{2} 1\{k = k'\}.
    \end{equation}

  \item For every $0 < \alpha < 1$, every $\sigma > 0$, and every
    $\mathsf{M} \in \Reals$ there exists a centered stationary SP $\{Y_{k}\}$
    whose R\'enyi rate exceeds $\mathsf{M}$ and that
    satisfies \eqref{eq:VarianceP}.
  \end{enumerate}
\end{proposition}
This proposition will be the key to the proof of
Theorem~\ref{thm:burgrenyi} ahead.

\subsection{Max R\'enyi Rate Subject to~\eqref{eq:Burg_cons} }

Given $\alpha_0, \ldots, \alpha_p\in\Reals$, consider the family of
all stochastic processes $X_1, X_2, \ldots$ satisfying~\eqref{eq:Burg_cons}.
Assume that the $(p+1)\times (p+1)$ matrix whose Row-$\ell$
Column-$m$ element is $\alpha_{|\ell-m|}$ is positive definite. Under
this assumption we have:
\begin{theorem}
  \label{thm:burgrenyi}
  The supremum of the order-$\alpha$ R\'enyi rate over all
  stochastic processes satisfying \eqref{eq:Burg_cons} is $+\infty$ for
  $0<\alpha<1$ and is equal to the Shannon rate of the $p$-th
  order Gauss-Markov process for $\alpha > 1$.
\end{theorem}


\section{Preliminaries}






\subsection{Weak Typicality}
\label{sec:weak_typ}

Given a density $f$ on $\set{S}$ of finite  Shannon
entropy 
\begin{equation}
  \label{eq:Finite_Shannon}
  -\infty < h(f) < \infty,
\end{equation}
a positive integer $n$, and some $\eps > 0$, we follow
\cite[Section~8.2]{cover2006elements} and denote by $\weaktyp(f)$ the
set of $\epsilon$-weakly-typical sequences of length $n$ with respect
to $f$:
\begin{IEEEeqnarray}{rCl}
  \IEEEeqnarraymulticol{3}{l}{%
    \weaktyp(f)
  }\nonumber\\*%
  & = & \biggl\{x_1^n \in \set{S}^n \colon 2^{-n(h(f)+\epsilon)}
  \leq \prod_{k=1}^n 
  f(x_k) \leq 2^{-n(h(f)-\epsilon)}\biggr\}.
  \nonumber\\*
  \label{eq:1}
\end{IEEEeqnarray}
By the AEP, if $X_{1}, \ldots, X_{n}$ are drawn IID according to some
such $f$, then the probability of $(X_{1}, \ldots, X_{n})$ being in
$\weaktyp(f)$ tends to $1$ as $n \to \infty$ (with $\eps$ held fixed)
\cite[Theorem 8.2.2]{cover2006elements}.

Given some measurable function $r\colon \set{S} \to \Reals$, some
density~$f$ that is supported on $\set{S}$ and that satisfies 
\begin{equation}
  \label{eq:r_int}
  \int_{\set{S}} f(x) \> |r(x)| \dd{x} < \infty,
\end{equation}
and given some $n \in \Naturals$ and $\eps > 0$, we
define
\begin{equation}
  \contyp{f} = \biggl\{x_1^n\in\set{S}^n \colon 
  \biggl| \frac{1}{n}\sum_{k=1}^n r(x_k) - \int_{\set{S}} f(x) \> r(x)
  \dd{x} \biggr| < \eps
  \biggr\}.
  \label{eq:4}
\end{equation}
By the Law of Large Numbers (LLN), if $X_{1}, \ldots, X_{n}$ are drawn
IID
according to some density $f$ that satisfies the above conditions, 
then the probability of $(X_{1}, \ldots, X_{n})$ being in
$\contyp{f}$ tends to $1$ as $n \to \infty$ (with $\eps$ held
fixed).

From the above observations on $\weaktyp(f)$ and $\contyp{f}$ we
conclude that if $X_{1}, \ldots, X_{n}$ are drawn IID according to
some density $f$ that is supported by $\set{S}$ and that
satisfies~\eqref{eq:Finite_Shannon} and~\eqref{eq:r_int},
then the probability of $(X_{1}, \ldots,
X_{n})$ being in the intersection $\weaktyp(f) \cap \contyp{f}$
tends to $1$ as $n \to \infty$. Thus, for all sufficiently large $n$,
\begin{IEEEeqnarray*}{rCl}
    1-\epsilon & \leq &  \int_{\weaktyp(f) \cap \contyp{f}}
    \prod_{k=1}^n f(x_k) \dd{x}^n
    \\ 
    & \leq &  \card{\weaktyp(f) \cap \contyp{f}}
    \> 2^{-n(h(f)-\epsilon)}, 
  \end{IEEEeqnarray*}
  where the second inequality holds by~\eqref{eq:1}.

  We thus conclude that if the support of $f$ is contained in
  $\set{S}$, the expectation of $|r(X)|$ under $f$ is finite, and
  $h(f)$ is defined and is finite, then
\begin{equation}
  \label{eq:Pre_card_lb}
  \card{\weaktyp(f) \cap \contyp{f}} \geq
  (1-\epsilon) \> 2^{n(h(f)-\epsilon)}, \quad \text{$n$ large.}
\end{equation}

\subsection{On the  R\'enyi Entropy of Mixtures}

The following lemma provides a lower bound on the R\'enyi entropy of a
mixture of densities in terms of the R\'enyi entropy of the individual
densities.
\begin{lemma}
  \label{lem:minimaljensen}
  Let $f_1, \ldots, f_p$ be probability density functions
  on~$\Reals^{n}$ and $q_1, \ldots, q_p \ge 0$ nonnegative numbers
  that sum to one. Let $f$ be the mixture density
  \begin{IEEEeqnarray*}{c}
    f(\vect{x}) = \sum_{\ell=1}^p q_{\ell} f_{\ell}(\vect{x}), \quad
    \vect{x}\in\Reals^n. 
  \end{IEEEeqnarray*}
  Then
  \begin{IEEEeqnarray*}{c}
    h_{\alpha}(f) \ge \min_{1\le \ell \le p} h_{\alpha}(f_{\ell}). 
  \end{IEEEeqnarray*}
\end{lemma}
\begin{proof}
  For $0 < \alpha < 1$ this follows by the concavity of R\'enyi
  entropy. Consider now $\alpha>1$:
  \begin{IEEEeqnarray*}{rCl}
    \log \int f^{\alpha}(\vect{x}) \dd \vect{x}
    & = &  \log \int \biggl( \sum_{\ell=1}^p q_{\ell}
    f_{\ell}(\vect{x})\biggr)^{\alpha} \dd\vect{x}
    \\
    & \leq &  \log \int \sum_{\ell=1}^p q_{\ell}
    f_{\ell}^{\alpha}(\vect{x}) \dd\vect{x}
    \\
    & = &  \log \left(  \sum_{\ell=1}^p q_{\ell} \int
      f_{\ell}^{\alpha}(\vect{x}) \dd\vect{x} \right)
    \\ 
    & \leq &  \log \max_{1 \le \ell \le p} \int
    f_{\ell}^{\alpha}(\vect{x}) \dd\vect{x}
    \\
    & = &  \max_{1 \le \ell \le p} \log \int
    f_{\ell}^{\alpha}(\vect{x}) \dd\vect{x},
  \end{IEEEeqnarray*}  
  from which the claim follows because $1/(1-\alpha)$ is
  negative. Here the first inequality follows from the convexity of
  the mapping $\xi \mapsto \xi^{\alpha}$ (for $\alpha > 1$), and the second
  inequality follows by upper-bounding the average by the maximum.
\end{proof}

We next turn to upper bounds.
\begin{lemma}
  \label{lem:MixtureUpper}
Consider the setup of Lemma~\ref{lem:minimaljensen}. 
\begin{enumerate}
\item If $\alpha > 1$ then
  \begin{equation}
    \label{eq:mix_up_1}
    h_{\alpha}(f) \le \min_{1\le \ell \le p} \Bigl\{\frac{\alpha}{1-\alpha}
    \log q_{\ell} + h_{\alpha}(f_{\ell}) \Bigr\}.
  \end{equation}
\item If $0 < \alpha < 1$ then
  \begin{equation}
    \label{eq:mix_up_2}
    h_{\alpha}(f) \le \frac{1}{1-\alpha} \log p + \max_{1 \le \ell \le p}
    h_{\alpha}(f_{\ell}). 
  \end{equation}
\end{enumerate}
\end{lemma}
\begin{proof}
  We begin with the case where $\alpha > 1$. Since the densities and
  weights are nonnegative,
  \begin{equation}
    \biggl( \sum_{\ell=1}^p q_{\ell}
    f_{\ell}(\vect{x})\biggr)^{\alpha} \geq \bigl( q_{\ell'}
    f_{\ell'}(\vect{x})\bigr)^{\alpha}, \quad \ell' \in \{1, \ldots, p\}.
  \end{equation}
  Integrating this inequality; taking logarithms, and dividing by
  $1-\alpha$ (which is negative) we obtain
  \begin{equation}
    h_{\alpha}(f) \le \frac{\alpha}{1-\alpha}
  \log q_{\ell'} + h_{\alpha}(f_{\ell'}), \quad \ell' \in \{1, \ldots, p\}.
  \end{equation}
  Since this holds for every $\ell' \in \{1, \ldots, p\}$, we can
  minimize over $\ell'$ to obtain~\eqref{eq:mix_up_1}.

  We next turn to the case where $0 < \alpha < 1$. 
  \begin{align*}
    \log \int \Bigl( \sum_{\ell=1}^{p} q_{\ell} f_{\ell}(\bfx)
    \Bigr)^{\alpha} \dd{\bfx} & \leq \log \int \max_{1 \leq \ell \leq p}
    f_{\ell}^{\alpha}(\bfx) \dd{\bfx} \\
    & \leq \log \int \sum_{\ell=1}^{p} f_{\ell}^{\alpha}(\bfx)
    \dd{\bfx} \\
    & = \log \sum_{\ell=1}^{p} \int f_{\ell}^{\alpha}(\bfx) \dd{\bfx} \\
    & \leq \log \biggl( p \max_{1 \leq \ell \leq p} \int
    f_{\ell}^{\alpha}(\bfx) \biggr) \dd{\bfx} \\
    & = \log p + \log \max_{1 \leq \ell \leq p} \int
    f_{\ell}^{\alpha}(\bfx) \dd{\bfx} \\
    & = \log p + \max_{1 \leq \ell \leq p} \log \int
    f_{\ell}^{\alpha}(\bfx) \dd{\bfx}.
  \end{align*}
  Dividing this inequality by $1-\alpha$ (positive) yields~\eqref{eq:mix_up_2}.
\end{proof}

\subsection{Bounded Densities}

\begin{proposition}
  \label{prop:RenyiOfBoundedDensity}
  If a density $f$ is bounded, and if $\alpha > 1$, then
  $h_{\alpha}(f) > -\infty$. 
\end{proposition}
\begin{proof}
  Let $f$ be a density that is upper-bounded by the constant
  $\mathsf{M}$ (which must therefore be positive), and suppose that
  $\alpha > 1$. In this case
  \begin{align*}
    f^{\alpha}(x) & = f^{\alpha - 1}(x) \> f(x) \\
    & \leq \mathsf{M}^{\alpha - 1} \> f(x),
  \end{align*}
  because $\xi \mapsto \xi^{\alpha - 1}$ is monotonically increasing
  when $\alpha > 1$. Integrating over $x$ we obtain
  \begin{equation*}
    \int f^{\alpha}(x) \dd{x} \leq \mathsf{M}^{\alpha - 1} < \infty.
  \end{equation*}
  Since $\alpha > 1$, this implies that 
  \begin{equation*}
    \frac{1}{1-\alpha} \log \intinf f^{\alpha}(x) \dd{x} > -\infty. \qedhere
  \end{equation*}
\end{proof}

The following proposition, which is proved in
Appendix~\ref{app:CanAssume}, demonstrates that $h^{\star}$ can be
approached by bounded densities.
\begin{proposition}
  \label{prop:CanAssume}
  Suppose that $\Gamma \in (\Gamma_{0}, \infty)$, where $\Gamma_{0}$
  satisfies~\eqref{block:amos}.
  Then for every $\delta > 0$ there exists some bounded density
  $f^{\star}$ supported by $\set{S}$ such that
  \begin{subequations}
  \begin{equation}
    \label{eq:14Prop}
    \int f^{\star}(x) \> r(x) \dd x < \Gamma + \delta,
  \end{equation}
  \begin{equation}
    \label{eq:13Prop}
    h(f^{\star}) > h^{\star}(\Gamma) - \delta.
  \end{equation}
  \end{subequations} 
\end{proposition}

\subsection{The Marginals of the Uniform Density on $\weaktyp(f)
  \cap \contyp{f}$}

\begin{lemma}
  \label{lem:RenThenFin}
  Let $f^{\star}$ be a density on $\set{S}$ having finite order-$\alpha$
  R\'enyi entropy
\begin{equation}
  h_{\alpha}(f^{\star}) > -\infty
\end{equation}
for some
\begin{equation}
  \label{eq:alpha_big1}
  \alpha > 1
\end{equation}
and satisfying~\eqref{eq:Finite_Shannon} and~\eqref{eq:r_int}.
For every $n \in \Naturals$, let
$(X_{1}, \ldots, X_{n})$ be drawn uniformly from the set
$\weaktyp(f^{\star}) \cap \contyp{f^{\star}}$, where $\eps$ is some fixed
positive number. Then for every
sufficiently large~$n$ the following holds: for any $\rho \in \{1,
\ldots, n\}$ the $\rho$-tuple $(X_{1}, \ldots, X_{\rho})$ has finite
order-$\alpha$ R\'enyi entropy
\begin{equation}
  \label{eq:first_rho_finite}
  h_{\alpha}(X_1, \ldots, X_{\rho}) > -\infty, \quad \Bigl( \rho \in \{1,
  \ldots, n\}, \; \alpha > 1 \Bigr).
\end{equation}
\end{lemma}
\begin{proof}
Denote the uniform density over $\weaktyp(f^{\star}) \cap \contyp{f^{\star}}$ by
$f_n$, and let $q_{n}$ be the product density
\begin{equation}
  \label{eq:def_qn}
  q_{n}(\bfx) = \prod_{k=1}^n f^{\star}(x_k), \quad \bfx \in \set{S}^{n}.
\end{equation}
Henceforth let $n$ be sufficiently large for \eqref{eq:Pre_card_lb} to
hold. Consequently,
\begin{IEEEeqnarray}{c}
  f_n(\vect{x}) \le \frac{1}{1-\eps} \> 2^{-n(h(f^{\star})-\epsilon)},
  \quad \bfx \in \set{S}^{n}.
\end{IEEEeqnarray}
Using this inequality and the definition in~\eqref{eq:1} of
$\weaktyp(f^{\star})$, we can upper-bound $f_{n}$ in terms of $q_{n}$
for tuples in $\weaktyp(f^{\star})$:
\begin{IEEEeqnarray}{c}
  \label{eq:ub_fn_using_qn}
  f_n(\vect{x}) \le \frac{1}{1-\eps} \> 2^{2n\epsilon} \> q_n(\vect{x}), 
  \quad \vect{x}\in \weaktyp(f^{\star}).
\end{IEEEeqnarray}
For every $\rho \in \{1, \ldots, n\}$ we can obtain the density
$f_n(x_1, \ldots, x_{\rho})$ of $(X_{1}, \ldots, X_{\rho})$ by
integrating $f_{n}(x_{1}, \ldots, x_{n})$ over $x_{\rho+1}, \ldots,
x_{n}$:
\begin{IEEEeqnarray*}{rCl}
  \IEEEeqnarraymulticol{3}{l}{%
    f_n(x_1, \ldots, x_{\rho}) 
  }\nonumber\\*
  & = & \int_{x_{\rho+1}, \ldots, x_n} f_n(\vect{x}) \I{\vect{x}\in\weaktyp(f^{\star})
    \cap \contyp{f^{\star}}} \dd x_{\rho+1} \cdots \dd x_n \\
  & \le & \frac{1}{1-\eps}\> 2^{2n\eps} 
   \!\! \int 
   \! q_n(\vect{x}) \I{\vect{x}\in\weaktyp(f^{\star})
     \cap \contyp{f^{\star}}} \dd x_{\rho+1} \cdots \dd x_n
   \\
   & \le & \frac{1}{1-\eps}\> 2^{2n\eps} 
   \!\! \int
   \! q_n(\vect{x}) \dd x_{\rho+1} \cdots  \dd x_n 
   \\ 
   & = & \frac{1}{1-\eps}\> 2^{2n\eps} 
   f^{\star}(x_1) \cdots f^{\star}(x_{\rho}), \quad x_1, \ldots,
   x_{\rho} \in \set{S},
   \IEEEyesnumber
   \IEEEeqnarraynumspace
   \label{eq:232}
\end{IEEEeqnarray*}
where $\I{\cdot}$ denotes the indicator function, and the first
inequality follows from~\eqref{eq:ub_fn_using_qn}; the second by
increasing the range of integration; and the final equality follows
from \eqref{eq:def_qn}.

Using~\eqref{eq:232} we can now lower-bound $h_{\alpha}(X_1, \ldots,
X_{\rho})$ as follows. If a density $f$ is upper-bounded by
$\mathsf{K} g$, where $g$ is some other density and $\mathsf{K}$ is
some positive constant, and if $\alpha > 1$, then
\begin{IEEEeqnarray*}{rCl}
  h_{\alpha}(f) & = & \frac{1}{1-\alpha} \log \int
  f^{\alpha}(\vect{x}) \dd \vect{x}
  \\
  & \ge & \frac{1}{1-\alpha} \log \int \mathsf{K}^{\alpha}
  g^{\alpha}(\vect{x}) \dd \vect{x}
  \\
  & = & \frac{\alpha}{1-\alpha}\log \mathsf{K} + h_{\alpha}(g),
  \IEEEyesnumber
  \IEEEeqnarraynumspace
  \label{eq:399}
\end{IEEEeqnarray*}
where the inequality 
holds because $\alpha > 1$ so the pre-log is negative. Using this and
\eqref{eq:232} we obtain
\begin{align*}
  h_{\alpha}(X_1, \ldots, X_{\rho}) & \ge  \frac{\alpha}{1-\alpha}
  \log \biggl(\frac{1}{1-\eps}\> 2^{2n\eps} \biggr) + \rho
  h_{\alpha}(f^{\star}) \\
  & > -\infty. \qedhere
\end{align*}
\end{proof}

\section{Proofs of Theorems~\ref{thm:AlphaBig}
  and~\ref{thm:AlphaSmall}}
\label{sec:Proofs1}

The following proposition is useful for stationarization.
\begin{proposition}
  \label{prop:stationarization}
  Let $f_{n}$ be some density on $\set{S}^{n}$ having order-$\alpha$
  R\'enyi entropy $h_{\alpha}(f_{n})$ and satisfying
  \begin{equation}
    \label{eq:constN}
    \sum_{k=1}^{n}\E{r(X_{k})} \leq n\Gamma, \quad (X_{1}, \ldots, X_{n}) \sim f_{n}.
  \end{equation}
  Then there exists a stationary
  SP $\DTSPgen{Z}{k}$ satisfying~\eqref{eq:AAconstSP}
  for which the following holds:
  \begin{itemize}
  \item If 
    \begin{multline}
      \label{eq:RhoRhoprimeBig}
      h_{\alpha}(X_{1}, \ldots, X_{\rho}), h_{\alpha}(X_{n-\rho'+1},
      \ldots, X_{n}) > -\infty, \\
      \quad \rho,\rho' \in \{1, \ldots, n-1\},
    \end{multline}
    whenever $(X_{1}, \ldots, X_{n}) \sim f_{n}$
    and 
    $\rho,\rho' \in \{1, \ldots, n-1\}$, then
    \begin{equation}
      \label{eq:liminfHZ}
      \varliminf_{m \to \infty} \frac{1}{m} h_{\alpha}(Z_{1}, \ldots, Z_{m}) \geq
      \frac{1}{n} h_{\alpha}(f_{n}).
    \end{equation}
  \item If 
    \begin{multline}
      \label{eq:RhoRhoprimeSmall}
      h_{\alpha}(X_{1}, \ldots, X_{\rho}), h_{\alpha}(X_{n-\rho'+1},
      \ldots, X_{n}) < +\infty, \\
      \quad \rho,\rho' \in \{1, \ldots, n-1\},
    \end{multline}
     whenever $(X_{1}, \ldots, X_{n}) \sim f_{n}$
    and 
    $\rho,\rho' \in \{1, \ldots, n-1\}$, then
    \begin{equation}
      \label{eq:limsupHZ}
      \varlimsup_{m \to \infty} \frac{1}{m} h_{\alpha}(Z_{1}, \ldots, Z_{m}) \leq
      \frac{1}{n} h_{\alpha}(f_{n}).
    \end{equation}
    \item And if both \eqref{eq:RhoRhoprimeBig} and
      \eqref{eq:RhoRhoprimeSmall} hold, then
      \begin{equation}
          \label{eq:limHZ} 
          \lim_{m \to \infty} \frac{1}{m} h_{\alpha}(Z_{1}, \ldots,
          Z_{m}) = 
          \frac{1}{n} h_{\alpha}(f_{n}).
      \end{equation}
  \end{itemize}
\end{proposition}
\begin{proof}
  Consider first the (nonstationary) SP $\{Y_{k}\}$
  that we construct by drawing
  \begin{equation*}
    \ldots, Y_{-n+1}^{0}, Y_{1}^{n}, Y_{n+1}^{2n}, \ldots \sim
    \text{IID $f_{n}$}.
  \end{equation*}
  To stationarize it, let $T$ be drawn uniformly over $\{0, \ldots,
  n-1\}$ independently of $\{Y_{k}\}$, and define the stationary SP
  \begin{equation}
    Z_{k} = Y_{k + T}, \quad k \in \Integers.
  \end{equation}
  It satisfies~\eqref{eq:AAconstSP}. Consider now any $m$ larger than
  $2n$, and express $Z_{1}^{m}$ in one of two different way depending
  on whether $T$ is zero or not. For $T=0$ 
  \begin{equation}
      \label{eq:Z1tomAXmas}
      Z_{1}^{m} = 
      \underbrace{Y_{1}^{n}, \ldots, 
        Y_{\tilde{\nu} n -n + 1}^{\tilde{\nu} n}}_{\text{$\tilde{\nu} = \lfloor m/n
          \rfloor$ $n$-tuples}}, \underbrace{Y_{\tilde{\nu} n + 1},
          \ldots, Y_{m}}_{\text{$\tilde{\rho} = m - n \lfloor m/n
          \rfloor$ terms}}
    \end{equation}
    where
  \begin{subequations}
    \label{block:Z1tomBXmas}
    \begin{equation}
      \label{eq:Z1tomBXmas2}
      \tilde{\nu} = \left\lfloor \frac{m}{n} \right\rfloor, 
      \end{equation}
      \begin{equation}
        \label{eq:Z1tomBXmas3}
        \tilde{\rho} = m - n \left\lfloor \frac{m}{n}
        \right\rfloor \in \{0, \ldots, n-1\}.
    \end{equation}
  \end{subequations}
  And for $T \in \{1, \ldots, n-1\}$
    \begin{multline}
      \label{eq:Z1tomA}
      Z_{1}^{m} = \\
      \underbrace{Y_{T+1}, \ldots, Y_{n}}_{\text{$\rho' = n-T$
          terms}}, \underbrace{Y_{n+1}^{2n}, \ldots, Y_{\nu n + 1}^{(\nu +
          1) n}}_{\text{$\nu$ $n$-tuples}}, \underbrace{Y_{(\nu+1)n + 1},
        \ldots, Y_{m+T}}_{\text{$\rho$ terms}}
    \end{multline}
    where
  \begin{subequations}
    \label{block:Z1tomB}
    \begin{equation}
      \label{eq:Z1tomB1}
      \rho' = n - T \in \{1, \ldots, n-1\},
    \end{equation}
    \begin{equation}
      \label{eq:Z1tomB2}
      \nu = \left\lfloor \frac{m - n + T}{n} \right\rfloor, 
      \end{equation}
      \begin{equation}
        \label{eq:Z1tomB3}
        \rho =
        m - n + T - n \left\lfloor \frac{m - n + T}{n} \right\rfloor
        \in \{0, \ldots, n-1\}.
    \end{equation}
  \end{subequations}

  Denote the density of $Z_{1}^{m}$ by $f_{\bfZ}$ and its
  conditional density given $T=t$ by $f_{\bfZ|T=t}$.  

  To establish~\eqref{eq:liminfHZ} we use
  Lemma~\ref{lem:minimaljensen}, which implies that 
  \begin{equation}
    \label{eq:Xmas100}
    h_{\alpha}\bigl(f_{\bfZ}\bigr) \geq \min_{0 \leq t \leq n-1}
    h_{\alpha}\bigl(f_{\bfZ|T = t}\bigr).
  \end{equation}
  
  To compute $h_{\alpha}\bigl(f_{\bfZ|T = 0}\bigr)$ we
  use~\eqref{eq:Z1tomAXmas} to obtain
  \begin{IEEEeqnarray}{rCl}
    \label{eq:lowerRenTeqZero1}
    h_{\alpha}\bigl(f_{\bfZ|T = 0}\bigr) & = & \left\lfloor \frac{m}{n} \right\rfloor
    h_{\alpha}(f_{n}) + h_{\alpha}(X_{1}, \ldots, X_{\tilde{\rho}}) \\
    & \geq & \left\lfloor \frac{m}{n} \right\rfloor
    h_{\alpha}(f_{n}) + 0 \wedge \min_{1 \leq \rho \leq n-1} \bigl\{ 
    h_{\alpha}(X_{1}, \ldots, X_{\rho}) \bigr\}.
    \IEEEeqnarraynumspace
    \label{eq:lowerRenTeqZero2}
  \end{IEEEeqnarray}
  where the second term on the RHS of~\eqref{eq:lowerRenTeqZero1}
  should be interpreted as zero when $\tilde{\rho}$ is
  zero, and where $a \wedge b$ denotes the minimum of $a$ and $b$.

  And to compute $h_{\alpha}\bigl(f_{\bfZ|T = t}\bigr)$ for $t \in
  \{1, \ldots, n-1\}$ we
    use~\eqref{eq:Z1tomA} to obtain
  \begin{IEEEeqnarray}{rCl}
    \label{eq:lb590}
      h_{\alpha}\bigl(f_{\bfZ|T = t}\bigr) 
    & = & 
    h_{\alpha}(X_{n-\rho'+1}, \ldots, X_{n}) \nonumber\\
    && +\: \left\lfloor
        \frac{m - n + t}{n} \right\rfloor h_{\alpha}(f_{n}) +
      h_{\alpha}(X_{1}, \ldots, X_{\rho}), \IEEEeqnarraynumspace
    \end{IEEEeqnarray}
    where $\rho,\rho'$ are obtained from \eqref{block:Z1tomB} by
    substituting $t$ for $T$, and the last term on the RHS should be
    interpreted as zero when $\rho$ is zero.

    It thus follows from~\eqref{eq:Xmas100},
    \eqref{eq:lowerRenTeqZero2}, \eqref{eq:lb590}, and the above
    interpretation that 
    \begin{IEEEeqnarray}{rCl}
      \label{eq:lb591}
      h_{\alpha}\bigl(f_{\bfZ}\bigr) 
      & \geq & \min_{1 \leq \rho' \leq n-1}
    \Bigl\{ h_{\alpha}(X_{n-\rho'+1}, \ldots, X_{n}) \Bigr\} \nonumber\\
    && +\: 0 \wedge \min_{1 \leq \rho \leq n-1} \Bigl\{h_{\alpha}(X_{1}, \ldots, X_{\rho}) \Bigr\}
    \nonumber \\
    && +\: \min_{0 \leq t \leq n-1} \biggl\{ 
    \left\lfloor
        \frac{m - n + t}{n} \right\rfloor h_{\alpha}(f_{n}) \biggr\}.
      \IEEEeqnarraynumspace
    \end{IEEEeqnarray}

    The first two terms do not depend on $m$ and are greater than
    $-\infty$ whenever \eqref{eq:RhoRhoprimeBig}
    holds. Dividing~\eqref{eq:lb591} by $m$ and letting $m$ tend to
    infinity (with $n$ held fixed), establishes~\eqref{eq:liminfHZ}.

  To establish~\eqref{eq:limsupHZ} we need an upper bound on
  $h_{\alpha}\bigl(f_{\bfZ}\bigr)$. Such a bound can be obtained from
  Lemma~\ref{lem:MixtureUpper}. The exact form of the bound depends on
  whether $\alpha$ exceeds $1$ or not. But either form leads
  to~\eqref{eq:limsupHZ} upon dividing by $m$ and letting it tend to
  infinity.
%

  To conclude the proof we note that~\eqref{eq:limHZ} follows
  from~\eqref{eq:limsupHZ} and~\eqref{eq:liminfHZ}.
\end{proof}

\begin{proof}[Proof of Theorem~\ref{thm:AlphaBig}]
  Since $h^{\star}(\cdot)$ is continuous on the ray $(\Gamma_{0},
  \infty)$, and since $\Gamma > \Gamma_{0}$ by the theorem's
  hypotheses, $h^{\star}(\cdot)$ is continuous at
  $\Gamma$. Consequently, we can find some $\Gamma'$ for which
  \begin{subequations}
    \begin{equation}
      \label{eq:Ch_GammaPrimeA}
      \Gamma' < \Gamma
    \end{equation}
    \begin{equation}
      \label{eq:Ch_GammaPrimeB}
      h^{\star}(\Gamma') > h^{\star}(\Gamma) - \tilde{\eps}.
    \end{equation}
  \end{subequations}
  These inequalities imply that we can find some $\delta > 0$ small enough so that
  \begin{subequations}
    \label{block:Ch_delta}
    \begin{equation}
      \label{eq:Ch_deltaA}
      \Gamma' + \delta < \Gamma
    \end{equation}
    \begin{equation}
      \label{eq:Ch_deltaB}
      h^{\star}(\Gamma') - \delta > h^{\star}(\Gamma) - \tilde{\eps}.
    \end{equation}
  \end{subequations}
  By Proposition~\ref{prop:CanAssume}, there exists some bounded
  density $f^{\star}$ supported by $\set{S}$ such that
\begin{subequations}
  \begin{equation}
    \label{eq:14PropUSE}
    \int f^{\star}(x) \> r(x) \dd x < \Gamma' + \delta,
  \end{equation}
  \begin{equation}
    \label{eq:13PropUSE}
    h(f^{\star}) >  h^{\star}(\Gamma') - \delta.
  \end{equation}
  Moreover, the boundedness of $f^{\star}$, the hypothesis that
  $\alpha > 1$, and Proposition~\ref{prop:RenyiOfBoundedDensity} imply
  that
  \begin{equation}
    \label{eq:15PropUSE}
    h_{\alpha}(f^{\star}) > -\infty.
  \end{equation}
  \end{subequations}
  These inequalities combine with~\eqref {block:Ch_delta} to imply
\begin{subequations}
  \label{block:Good_fstar}
  \begin{equation}
    \label{eq:Good_fstarA}
    \int f^{\star}(x) \> r(x) \dd x < \Gamma
  \end{equation}
  \begin{equation}
    \label{eq:Good_fstarB}
    h(f^{\star}) >  h^{\star}(\Gamma) - \tilde{\eps}.
  \end{equation}
  \end{subequations}
  We can hence choose $\eps > 0$ small enough so that
\begin{subequations}
  \label{block:Good_eps}
  \begin{equation}
    \label{eq:Good_epsA}
    \int f^{\star}(x) \> r(x) \dd x < \Gamma - \eps
  \end{equation}
  \begin{equation}
    \label{eq:Good_epsB}
    h(f^{\star}) >  h^{\star}(\Gamma) - \tilde{\eps} + \eps.
  \end{equation}
  \end{subequations}
  
  Let $f_{n}$ be the uniform density over 
  \begin{equation*}
    \weaktyp(f^{\star}) \cap \contyp{f^{\star}}.
  \end{equation*}

  The cost of $f_{n}$ can be bounded by noting that its support is
  contained in $\contyp{f^{\star}}$, and
  \begin{align*}
    x_{1}^{n} \in \contyp{f^{\star}} & \implies
    \frac{1}{n}\sum_{k=1}^n r(x_k) < \int f^{\star}(x) \> r(x) \dd{x}
    + \eps \\
    & \implies \frac{1}{n}\sum_{k=1}^n r(x_k) < \Gamma,
  \end{align*}
  where the second implication follows from~\eqref{eq:Good_epsA}. Thus,
  \begin{equation}
    \label{eq:SP_sum_const}
    \int_{\set{S}^{n}} f_{n}(\bfx) \sum_{i=1}^{n} r(x_{i}) \dd{\bfx} \leq n \Gamma.
  \end{equation}

  To lower-bound its R\'enyi entropy, we note that by the LLN (in combination
  with \eqref{eq:Good_epsA}) and the AEP (see Section~\ref{sec:weak_typ})
  \begin{equation}
    \label{eq:Stat_card_lb}
    \card{\weaktyp(f^{\star}) \cap \contyp{f^{\star}}} \geq
    (1-\epsilon) \> 2^{n(h(f^{\star})-\epsilon)}, \quad \text{$n$ large.}
  \end{equation}
  Consequently,
  \begin{equation*}
    h_{\alpha}(f_n) \geq n \bigl(h(f^{\star})-\epsilon\bigr) + \log(1-\epsilon)
    \quad \text{$n$ large}, 
  \end{equation*}
or, upon dividing by $n$,
\begin{equation}
  \label{eq:Witwe10}
    \frac{1}{n} h_{\alpha}(f_n) \geq h(f^{\star})-\eps + \frac{1}{n}\log(1-\epsilon)
  \end{equation}
  for all sufficiently large $n$.  We now choose $n$ large enough so
  that not only will \eqref{eq:Witwe10} hold but also its RHS
  satisfy
  \begin{equation*}
    h(f^{\star})-\eps + \frac{1}{n}\log(1-\epsilon) >
    h^{\star}(\Gamma) - \tilde{\eps}.
  \end{equation*}
  (This is possible by~\eqref{eq:Good_epsB}.) For this $n$ we thus
  have 
  \begin{equation}
    \label{eq:Good_n}
    \frac{1}{n} h_{\alpha}(f_n) >  h^{\star}(\Gamma) - \tilde{\eps}.
  \end{equation}

  The inequalities \eqref{eq:Good_n} and \eqref{eq:SP_sum_const}
  indicate that $f_{n}$ is a good candidate for the application of
  Proposition~\ref{prop:stationarization}. We hence proceed to check
  its hypotheses.

  By Lemma~\ref{lem:RenThenFin} and \eqref{eq:15PropUSE}, if $X_{1},
  \ldots, X_{n} \sim f_{n}$ then
  \begin{equation}
    h_{\alpha}(X_{1}, \ldots, X_{\rho}) > -\infty, 
      \quad \rho \in \{1, \ldots, n-1\},
  \end{equation}
  and, since $f_{n}$ is permutation invariant, we also infer
\begin{equation}
    h_{\alpha}(X_{n-\rho'+1},
      \ldots, X_{n}) > -\infty, 
      \quad \rho' \in \{1, \ldots, n-1\}
  \end{equation}
  so~\eqref{eq:RhoRhoprimeBig} holds. And, since $\alpha > 1$, it
  follows from~\eqref{eq:h_le_plus} that~\eqref{eq:RhoRhoprimeSmall}
  also holds. We can thus apply
  Proposition~\ref{prop:stationarization} to conclude the proof.
\end{proof}

\begin{proof}[Proof of Theorem~\ref{thm:AlphaSmall}]
  We first prove the theorem when $\card{\set{S}} = \infty$. We
  distinguish between two cases. The first case, which is the case
  with which we begin, is when there exists some $n \in \Naturals$ and
  a density $f_{n}^{\star}$ on $X_{1}, \ldots, X_{n}$ such that
  \begin{equation}
    \label{eq:CONSn66}
    \Prv{X_{i} \in \set{S}} = 1, \quad \E{r(X_{i})} \leq \Gamma, \quad
    i \in \{1, \ldots, n\}
  \end{equation}
  and
  \begin{equation}
    h_{\alpha}(X_{1}, \ldots, X_{n}) = + \infty.
  \end{equation}
  To apply Proposition~\ref{prop:stationarization} to this density, we note
  that, since $0 < \alpha < 1$, Inequality~\eqref{eq:h_ge_minus}
  implies~\eqref{eq:RhoRhoprimeBig}, and the proposition thus
  guarantees the existence of a stationary SP $\DTSPgen{Z}{k}$ satisfying~\eqref{eq:AAconstSP}
  and~\eqref{eq:liminfHZ} so
  \begin{equation}
    \lim_{m \to \infty} \frac{1}{m} h_{\alpha}(Z_{1}, \ldots,
          Z_{m}) = +\infty.
  \end{equation}
  This concludes the proof for the case at hand. 
  
  We next turn to the second case where $|\set{S}|$ is still infinite,
  but any tuple whose components satisfy the constraints has 
  R\'enyi entropy smaller than $\infty$:
  \begin{multline}
    \label{eq:Eilat_second_case}
    \biggl( \Prv{X_{i} \in \set{S}} = 1, \quad \E{r(X_{i})} \leq \Gamma, \quad
    i \in \{\nu_{1}, \ldots, \nu_{2}\} \biggr) \\
    \implies \biggl( h_{\alpha}(X_{\nu_{1}}, \ldots, X_{\nu_{2}}) <  \infty \biggr).
  \end{multline}
  
  Since $|\set{S}|$ is infinite, it follows from
  Proposition~\ref{prop:HstarProperties} that $h^{\star}(\Gamma) \to
  \infty$ as $\Gamma \to \infty$. Consequently, there exists some
  $\Gamma_{1}$ such that
  \begin{equation}
    \label{eq:Eilat10}
    h^{\star}(\Gamma_{1}) > \mathsf{M}.
  \end{equation}
  Since $h^{\star}$ is monotonic, there is no loss in generality in
  assuming, as we shall, that
  \begin{equation}
    \label{eq:Eilat20}
    \Gamma_{1}  > \Gamma.
  \end{equation}
  Let $\eps \in (0,1)$ be small enough so that
  \begin{equation}
    \label{eq:Eilat30}
    h^{\star}(\Gamma_{1}) > \mathsf{M} + 3 \eps
  \end{equation}
  \begin{equation}
    \label{eq:Eilat40}
    \Gamma_{0} + \eps < \Gamma < \Gamma_{1} - \eps. 
  \end{equation}
  Let the densities $f^{(0)}$ and $f^{(1)}$ be within $\eps$ of
  achieving $h^{\star}(\Gamma_0)$ and $h^{\star}(\Gamma_1)$ in the
  sense that their support is contained in $\set{S}$ and
  \begin{multline}
    \label{eq:Eilat_eps_achieve}
    \biggl( \int_{\set{S}} f^{(\ell)}(x) \> r(x) \dd{x} \leq \Gamma_{\ell},
    \quad h \bigl( f^{(\ell)} \bigr) > 
    h^{\star}(\Gamma_{\ell}) - \epsilon \biggr), \\ \ell\in\{0,1\}.
  \end{multline}
  For every $n \in \Naturals$, define
  \begin{equation}
    \label{eq:Eilat_def_Sell}
    \set{S}_{\ell} = \weaktyp\bigl(f^{(\ell)}\bigr) \cap
    \contyp{f^{(\ell)}},
    \quad \ell\in\{0,1\}.
  \end{equation}
  It follows from the LLN and AEP that, for all sufficiently large~$n$,
  \begin{equation}
    \label{eq:Eilat9_card_lb}
    \card{\set{S}_{\ell}} \geq
    (1-\epsilon) \> 2^{n(h(f^{(\ell)})-\epsilon)}, \quad \ell \in \{0,1\}. 
  \end{equation}
  Assume now that $n$ is large enough for this to hold.
  Let $\delta > 0$ be small enough so that
  \begin{equation}
    \label{eq:Eilat100}
    (1-\delta) \, (\Gamma_0 + \eps) + \delta \, (\Gamma_1 + \eps) \leq \Gamma. 
  \end{equation}
  (Such a $\delta$ can be found in view of~\eqref{eq:Eilat40}.)

  Consider now the mixture density
  \begin{equation}
    \label{eq:mixture_density}
    f_n(x_1^n) = (1-\delta)\frac{1}{\card{\set{S}_0}} \I{x_1^n \in
    \set{S}_0} + \delta  \frac{1}{\card{\set{S}_1}} \I{x_1^n\in
    \set{S}_1}. 
  \end{equation}
  Let $X_1^n$ be of density $f_n$. Using~\eqref{eq:Eilat100} and an
  argument similar to the one leading to~\eqref{eq:SP_sum_const} we
  obtain
  \begin{equation}
    \label{eq:Eilat110}
    \sum_{k=1}^{n} \E{r(X_{k})} \leq n \Gamma. 
  \end{equation}
  In fact, the permutation invariance of $f_{n}$ implies the stronger
  statement
  \begin{equation}
    \label{eq:Eilat120}
    \E{r(X_{k})} \leq \Gamma, \quad k = 1, \ldots, n.
  \end{equation}

  We next lower-bound $h_{\alpha}(X_1^n)$. To this end, we first argue
  that the sets $\set{S}_0$ and $\set{S}_1$ are disjoint. To see this,
  note that by the definition of the sets $\contyp{f^{(0)}}$,
  $\contyp{f^{(1)}}$ and by~\eqref{eq:Eilat_eps_achieve}
\begin{align}
    x_{1}^{n} \in \contyp{f^{(0)}} & \implies
    \frac{1}{n}\sum_{k=1}^n r(x_k) < \int f^{(0)}(x) \> r(x) \dd{x}
    + \eps \nonumber \\
    & \implies \frac{1}{n}\sum_{k=1}^n r(x_k) < \Gamma_{0} + \eps, \label{eq:skin10}
  \end{align}
  and
  \begin{align}
    x_{1}^{n} \in \contyp{f^{(1)}} & \implies
    \frac{1}{n}\sum_{k=1}^n r(x_k) >  \int f^{(1)}(x) \> r(x) \dd{x}
    - \eps \nonumber \\
    & \implies \frac{1}{n}\sum_{k=1}^n r(x_k) > \Gamma_{1} - \eps, \label{eq:skin20}
  \end{align}
  From~\eqref{eq:Eilat40}, \eqref{eq:skin10}, and~\eqref{eq:skin20} we
  now conclude that $\contyp{f^{(0)}}$ and $\contyp{f^{(1)}}$ are
  disjoint and hence also $\set{S}_0$ and $\set{S}_1$. 

  Having established that $\set{S}_0$ and $\set{S}_1$ are disjoint, we
  can now compute $h_{\alpha}(f_{n})$ directly to obtain:
  \begin{IEEEeqnarray}{rCl}
    \frac{h_{\alpha}(X_1^n)}{n}
    & = & \frac{1}{n(1-\alpha)} \log \Bigl((1-\delta)^{\alpha}
    \card{\set{S}_0}^{1-\alpha}+\delta^{\alpha}
    \card{\set{S}_1}^{1-\alpha}\Bigr)
    \nonumber\\
    & \geq & \frac{1}{n(1-\alpha)} \log \Bigl(\delta^{\alpha}
    \card{\set{S}_1}^{1-\alpha}\Bigr). \label{eq:Eilat117}
  \end{IEEEeqnarray}
  From this, \eqref{eq:Eilat9_card_lb}, \eqref{eq:Eilat_eps_achieve},
  and \eqref{eq:Eilat30} it now follows that we can find some
  sufficiently large $n$ for which
  \begin{equation}
    \label{eq:Eilat200}
    \frac{h_{\alpha}(X_1^n)}{n} > \mathsf{M}.
  \end{equation}
  To apply Proposition~\ref{prop:stationarization} we note
  that~\eqref{eq:Eilat120} and \eqref{eq:Eilat_second_case} imply that
  \eqref{eq:RhoRhoprimeSmall} holds. And the fact that $\alpha \in
  (0,1)$ implies by \eqref{eq:h_ge_minus} that
  \eqref{eq:RhoRhoprimeBig} holds. Hence, by the proposition, there
  exists a stationary SP satisfying the constraints and whose R\'eny
  rate is $n^{-1} h_{\alpha}(X_1^n)$ and thus exceeds
$\mathsf{M}$. This concludes the proof when $\card{\set{S}} =
\infty$. 

The proof when $\card{\set{S}} < \infty$ is very similar. In fact, it
is a bit simpler because $\card{\set{S}} < \infty$
implies~\eqref{eq:Eilat_second_case}. We begin the proof by noting
that, since $\card{\set{S}} < \infty$,
Proposition~\ref{prop:HstarProperties} implies that $h^{\star}(\Gamma)
\to \log \card{\set{S}}$ as $\Gamma \to \infty$. Consequently, there
exists some~$\Gamma_{1}$ such that
  \begin{equation}
    \label{eq:FiniteEilat10}
    h^{\star}(\Gamma_{1}) > \log \card{\set{S}} - \tilde{\eps}.
  \end{equation}
  Replacing $\mathsf{M}$ with $\log \card{\set{S}} - \tilde{\eps}$ in
  the derivation that leads from \eqref{eq:Eilat10} to
  \eqref{eq:Eilat200}, we obtain a density $f_{n}$ for which
  \begin{equation}
    \label{eq:FiniteEilat200}
    \frac{h_{\alpha}(X_1^n)}{n} > \log \card{\set{S}} - \tilde{\eps}.
  \end{equation}
  The result then follows from Proposition~\ref{prop:stationarization}
  by noting that the LHS of \eqref{eq:RhoRhoprimeSmall} is upper
  bounded by $n \log \card{\set{S}}$ and by noting
  that~\eqref{eq:RhoRhoprimeBig} holds by~\eqref{eq:h_ge_minus}
  because $0 < \alpha < 1$.
\end{proof}

\section{Proof of Theorem~\ref{thm:burgrenyi}}
\label{sec:proofburgrenyi}

\begin{proof}[Proof of Theorem~\ref{thm:burgrenyi}]
  Recall the assumption that the $(p+1)\times (p+1)$ matrix whose
  Row-$\ell$ Column-$m$ element is $\alpha_{|\ell-m|}$ is positive
  definite. This
implies \cite{pourahmadi2001foundations} that there exist constants
$a_1, \ldots, a_p, \sigma^2$ and a $p\times p$ positive definite
matrix $\mat{K}_p$ such that the following holds:\footnote{The
  Row-$\ell$ Column-$m$ element of the matrix $\mat{K}_{p}$ is
  $\alpha_{|\ell-m|}$. This matrix is thus the result of deleting the
  last column and last row of the $(p+1)\times(p+1)$ matrix that we
  assumed was positive definite.}  if the random $p$-vector $(W_{1-p}, \ldots,
W_{0})$ is of second-moment matrix $\mat{K}_p$ (not necessarily
centered) and if $\{Z_i\}_{i=1}^{\infty}$ are independent of
$(W_{1-p}, \ldots, W_{0})$ with
\begin{IEEEeqnarray}{rCl"l}
  \IEEEyesnumber
  \label{eq:zcons}
  \E{Z_i} & = & 0, & i\in\Naturals,
  \IEEEyessubnumber
  \label{eq:zcons_a}
  \\
  \E{Z_i Z_j} & = & \sigma^2 \I{i=j}, & i,j\in\Naturals,
  \IEEEyessubnumber
  \label{eq:zcons_b}  
\end{IEEEeqnarray}
then the process defined inductively via
\begin{IEEEeqnarray}{c}
  X_i = \sum_{k=1}^p a_i X_{i-k} + Z_i, \quad i\in\Naturals 
  \label{eq:a}
\end{IEEEeqnarray}
with the initialization
\begin{IEEEeqnarray}{c}
  (X_{1-p}, \ldots, X_{0}) = (W_{1-p}, \ldots, W_{0})
  \label{eq:b}
\end{IEEEeqnarray}
satisfies the constraints~\eqref{eq:Burg_cons}.

(By Burg's maximum entropy theorem
\cite[Theorem~12.6.1]{cover2006elements}, of all stochastic processes
satisfying \eqref{eq:Burg_cons} the one of highest Shannon rate is the
$p$-th order Gauss-Markov process. It is obtained when $(W_{1-p},
\ldots, W_{0})$ is a centered Gaussian and $\{Z_i\}$ are IID
$\sim\Normal{0}{\sigma^2}$. Its Shannon entropy rate is $(1/2) \log
(2\pi e \sigma^2).$)

We first consider the case where $\alpha > 1$. Let $a_1, \ldots,
  a_p, \sigma^2$ and $\mat{K}_p$ be as above, and let $\eps > 0$ be
  arbitrarily small. By Proposition~\ref{cor:SecondMoment} there
  exists a SP $\{Z_i\}$ such that \eqref{eq:zcons}
  holds and such that
  \begin{IEEEeqnarray}{c}
    \lim_{n\to\infty} \frac{1}{n} h_{\alpha}(Z_1, \ldots, Z_n) 
    \ge \frac{1}{2}\log(2\pi e \sigma^2) - \eps.
    \label{eq:df1}    
  \end{IEEEeqnarray}
  The matrix $\mat{K}_p$ is positive definite, so by the spectral
  representation theorem we can find vectors $\vect{w}_1, \ldots,
  \vect{w}_p\in\Reals^p$ and constants $q_1, \ldots, q_p >0$ with
  $q_1+\cdots+q_p=1$ such that
  \begin{IEEEeqnarray}{c}
    \mat{K}_p = \sum_{\ell=1}^p q_{\ell}
    \vect{w}_{\ell} \trans{\vect{w}_{\ell}}.
    \label{eq:kpassum}    
  \end{IEEEeqnarray}
  (The vectors are eigenvectors of $\mat{K}_p$, and the constants
  $q_1, \ldots, q_p$ are the scaled eigenvalues of $\mat{K}_p$.) Draw
  the random vector $\vect{W}$ independently of $\{Z_i\}$ with 
  \begin{IEEEeqnarray*}{c}
    \Prv{\vect{W}=\vect{w}_{\ell}} = q_{\ell},
  \end{IEEEeqnarray*}
  so that, by \eqref{eq:kpassum}, 
  \begin{IEEEeqnarray*}{c}
    \E{\vect{W}\trans{\vect{W}}} = \mat{K}_p.
  \end{IEEEeqnarray*}
  Construct now the stochastic process $\{X_i\}$ using \eqref{eq:a}
  initialized with $\trans{(X_{1-p}, \ldots, X_{0})}$ being set to
  $\vect{W}$.

  The resulting SP thus
  satisfies~\eqref{eq:Burg_cons}. We next study its R\'enyi
  rate. To that end, we study the R\'enyi entropy of the vector
  $X_1^n$. Let $f_{\vect{X}}$ denote its density, and let
  $f_{\vect{X}|\vect{w}_{\ell}}$ denote its conditional density given
  $\vect{W}=\vect{w}_{\ell}$, so
  \begin{IEEEeqnarray*}{c}
    f_{\vect{X}}(\vect{x}) = \sum_{\ell=1}^p q_{\ell}
    f_{\vect{X}|\vect{w}_{\ell}}(\vect{x}), \quad \vect{x}\in\Reals^n.
  \end{IEEEeqnarray*}
  Consequently, by Lemma~\ref{lem:minimaljensen},
  \begin{IEEEeqnarray}{c}
    h_{\alpha}(f_{\vect{X}}) \ge \min_{1\le\ell\le p}
    h_{\alpha}(f_{\vect{X}|\vect{w}_{\ell}}),
    \label{eq:lbmin}
  \end{IEEEeqnarray}
  and by Lemma~\ref{lem:MixtureUpper}
\begin{equation}
    \label{eq:ubBurg}
    h_{\alpha}(f_{\vect{X}}) \le \min_{1\le \ell \le p} \Bigl\{\frac{\alpha}{1-\alpha}
    \log q_{\ell} +  h_{\alpha}(f_{\vect{X}|\vect{w}_{\ell}}) \Bigr\}.
  \end{equation}
  We next study $h_{\alpha}(f_{\vect{X}|\vect{w}_{\ell}})$ for any
  given $\ell\in\{1, \ldots, p\}$. Recalling that $\vect{W}$ and
  $\{Z_i\}$ are independent, we conclude that, conditional on
  $\vect{W}=\vect{w}_{\ell}$, the random variables $X_1, \ldots, X_n$
  are generated inductively via \eqref{eq:a} with the
  initialization 
  \begin{IEEEeqnarray*}{c}
    \trans{(X_{1-p}, \ldots, X_{0})} = \vect{w}_{\ell}.
  \end{IEEEeqnarray*}
  Conditionally on $\vect{W}=\vect{w}_{\ell}$, the random variables
  $X_1, \ldots, X_n$ are thus an affine transformation of $Z_1, \ldots,
  Z_n$. The transformation is of unit Jacobian (because the
  partial-derivatives matrix has $1$'s on the diagonal and $0$'s on
  the upper triangle), and thus
  \begin{equation}
    h_{\alpha}(f_{\vect{X}|\vect{w}_{\ell}}) = h_{\alpha}(Z_1, \ldots,
    Z_n), \quad \ell\in\{1, \ldots, p\}.
  \end{equation}
  From this, \eqref{eq:lbmin}, and~\eqref{eq:ubBurg} it follows that
  \begin{IEEEeqnarray*}{c}
    h_{\alpha}(Z_1^{n}) \leq h_{\alpha}(f_{\vect{X}}) \leq \min_{1\le \ell \le p} \Bigl\{\frac{\alpha}{1-\alpha}
    \log q_{\ell} \Bigr\} + h_{\alpha}(Z_1^{n}).
  \end{IEEEeqnarray*}
  Dividing by $n$ and using \eqref{eq:df1} establishes the result.

  We next turn to the case $0 < \alpha < 1$. For every $\mathsf{M}>0$
  arbitrarily large, we use Proposition~\ref{cor:SecondMoment} to
  construct $\{Z_i\}$ as above but with 
  \begin{IEEEeqnarray*}{c}
    \lim_{n\to\infty} \frac{1}{n} h_{\alpha}(Z_1, \ldots, Z_n) 
    \ge \mathsf{M}.
  \end{IEEEeqnarray*}
  The proof continues as for the case where $\alpha$ exceeds one.
\end{proof}

\section{Discussion}
\subsection{On Theorem~\ref{thm:AlphaBig}}
As the following heuristic argument demonstrates, one has to walk a
fine line in order to achieve the supremum promised in
Theorem~\ref{thm:AlphaBig}. To see why, let us focus on the case where
$h^{\star}(\cdot)$ is strictly increasing and where there exist
real constants $\lambda_{0},\lambda_{1} \in \Reals$ for which the function
$f^{\star}(x) = \exp{\bigl(\lambda_{0} + \lambda_{1} r(x) \bigr)} \I{x
  \in \set{S}}$ is a density achieving $h^{\star}(\Gamma)$.
For any other density $g$ supported on $\set{S}$ and satisfying
  \begin{equation}
   \label{eq:Dis_g_cons}
    \int_{\set{S}} g(x) \, r(x) \dd{x} = \Gamma
  \end{equation}
  we then have (as in the proof of
  \cite[Theorem~12.1.1]{cover2006elements})
  \begin{align}
    h(g) & =  h(f^{\star}) - D(g\|f^{\star}) \\
    & = h^{\star}(\Gamma) - D(g\|f^{\star}).
  \end{align}
  
  Using this and~\eqref{eq:up_by_sum_Shannon} we thus obtain that if
  $\DTSPgen{Z}{k}$ is a stationary SP and if $f_{Z}$ is the
  density of $Z_{1}$ and 
  \begin{equation}
   \label{eq:Dis_fZ_cons}
    \int_{\set{S}} f_{Z}(x) \, r(x) \dd{x} = \Gamma,
  \end{equation}
  then
  \begin{equation}
    \label{eq:fine_line}
    h_{\alpha}(\DTSPgen{Z}{k}) \leq h^{\star}(\Gamma) - D(f_{Z}\|f^{\star}),
    \quad \alpha > 1.
  \end{equation}
  Thus, for $h_{\alpha}(\DTSPgen{Z}{k})$ to be close to
  $h^{\star}(\Gamma)$, the density of $Z_{1}$ must be ``close'' (in
  relative-entropy) to $f^{\star}$.\footnote{We are ignoring here the
    fact that one might consider approaching the supremum
    with~\eqref{eq:Dis_fZ_cons} only being an inequality.}  We can
  repeat this argument for the joint density of $Z_{1},Z_{2}$ to infer
  that $Z_{1}$ and $Z_{2}$ must be ``nearly independent'' with each
  being of density ``nearly'' $f^{\star}$. More generally, for every
  fixed $m \in \Naturals$ the joint density of $Z_{1}, \ldots, Z_{m}$
  must be nearly of a product form. But, of course choosing
  $\DTSPgen{Z}{k}$ IID will not work, because this choice would lead
  to a R\'enyi rate equal to $h_{\alpha}(f_{Z_{1}})$, which is
  typically smaller than $h(Z_{1})$ (see~\eqref{eq:Ren_le_Shannon}).

\subsection{On Theorem~\ref{thm:burgrenyi}}
Theorem~\ref{thm:burgrenyi} has bearing on the spectral estimation
problem, i.e., the problem of extrapolating the values of the
autocovariance sequence from its first $p+1$ values. One approach is
to choose the extrapolated sequence to be the autocovariance sequence
of the stochastic process that---among all stochastic processes that
have an autocovariance sequence that starts with these $p+1$
values---maximizes the Shannon rate, namely the $p$-th order
Gauss-Markov process (Burg's theorem).  

A different approach might be to choose some $\alpha > 1$ and to
replace the maximization of the Shannon rate with that of the
order-$\alpha$ R\'enyi rate. As we next argue,
Theorem~\ref{thm:burgrenyi} shows that this would result in the same
extrapolated sequence. Indeed, inspecting the proof of the theorem we
see that the stochastic process $\{X_{i}\}$ that we constructed, while
not a Gauss-Markov process, has the same autocovariance sequence as
the $p$-th order Gauss-Markov process that satisfies the
constraints. And, for $\alpha > 1$ the supremum can only be achieved
by a stochastic process of this autocovariance sequence: for any other
autocovariance function the R\'enyi rate is upper bounded by the
Shannon rate (because $\alpha > 1$), and the latter is upper bounded
by the Shannon rate of the Gaussian process, which, unless the
autocovariance sequence is that of the $p$-th order Gauss-Markov
process, is strictly smaller than the supremum (Burg's theorem).

\appendix

\section{Proof of Proposition~\ref{prop:CanAssume}}
\label{app:CanAssume}

In this appendix we present two lemmas, which we then use to prove
Proposition~\ref{prop:CanAssume} on approaching $h^{\star}(\Gamma)$
using bounded densities.

\begin{lemma}
  \label{lem:anotherlemma}
  Let $f$ be a density supported by $\set{S}$ for which $h(f)$ is defined;
  \begin{equation}
    \label{eq:cost_integrable}
    \int f(x) \> |r(x)| \dd x < \infty;
  \end{equation}
 and for which 
  \begin{equation}
    \label{eq:cons13}
    \int f(x) \> r(x) \dd x \le \Gamma
  \end{equation}
  for some $\Gamma \in \Reals$. Then for every $\delta > 0$ there
  exists a density~$\tilde{f}$ that  is bounded, supported by $\set{S}$,
  and that satisfies
  \begin{equation}
    \label{eq:14}
    \int \tilde{f}(x) \> r(x) \dd x \le \Gamma + \delta
  \end{equation}
  and 
  \begin{equation}
    \label{eq:13}
    h(\tilde{f}) \ge h(f) - \delta.
  \end{equation}
\end{lemma}
\begin{proof}
  Let $0 < \eps < 1$ be fixed (small), with its choice specified
  later. It follows from~\eqref{eq:cost_integrable}
  and the MCT that there exists some $\mathsf{M}_1$ sufficiently large
  so that
  \begin{equation*}
    \int \Bigl( f(x) - \bigl(f(x) \wedge \mathsf{M}_1 \bigr) \Bigr) \>
    |r(x)| \dd x  < \eps, 
  \end{equation*}
  where we recall that $a\wedge b$ stands for $\min\{a,b\}$. Since the density $f$
  integrates to 1, we can find some $\mathsf{M}_2$ sufficiently large
  so that 
  \begin{equation*}
    \int \bigl(f(x) \wedge \mathsf{M}_2 \bigr) \dd x > 1-\eps.
  \end{equation*}
  Define now 
  \begin{equation}
    \label{eq:defM}
    \mathsf{M} = \max\{1, \mathsf{M}_1, \mathsf{M}_2 \}.
  \end{equation}
  For this $\mathsf{M}$ we have:
  \begin{IEEEeqnarray}{c}
    \int \bigl(f(x) \wedge \mathsf{M} \bigr) \dd x > 1-\eps,
    \IEEEyesnumber\IEEEyessubnumber
    \label{eq:15a}
    \\
    \int \Bigl( f(x) - \bigl( f(x) \wedge \mathsf{M} \bigr) \Bigr) \>
    |r(x)| \dd x   < \eps, 
    \IEEEyessubnumber
    \label{eq:15b}
    \\
    \Bigl( f(x) \ge 1 \Bigr) \implies \Bigl( f(x) \wedge
    \mathsf{M} \ge 1 \Bigr).
    \IEEEyessubnumber
    \label{eq:15c}
  \end{IEEEeqnarray}
  Consider now the bounded density 
  \begin{subequations}
    \label{block:def_f_tilde}
  \begin{equation}
    \label{eq:def_f_tildeA}
    \tilde{f}(x) = \frac{1}{\beta} \bigl( f(x) \wedge \mathsf{M} \bigr)
  \end{equation}
  where 
  \begin{equation}
    \beta = \int \bigl( f(\tilde{x}) \wedge \mathsf{M}\bigr) \dd \tilde{x}.
  \end{equation}
  \end{subequations}
  Note that because $f(x) \wedge \mathsf{M}$ is upper-bounded by
  $f(x)$, which integrates to one, and because of~\eqref{eq:15a}
  \begin{equation}
    \label{eq:bound_onBeta}
    1 - \eps \leq \beta \leq 1,
  \end{equation}
  so
  \begin{equation}
    \bigl(f(x) \wedge \mathsf{M} \bigr) 
    \le \tilde{f}(x) \le 
    \frac{1}{1-\eps} \> \bigl( f(x) \wedge \mathsf{M} \bigr).
  \end{equation}
  Moreover, $\tilde{f}$ is supported by $\set{S}$.

  Given $\delta > 0$ we next show that by choosing $\epsilon$
  sufficiently small we can guarantee that both \eqref{eq:14} and
  \eqref{eq:13} hold. Be begin with the former.  Starting with
  \eqref{eq:def_f_tildeA} we have
  \begin{IEEEeqnarray*}{rCl}
    \IEEEeqnarraymulticol{3}{l}{%
      \int \tilde{f}(x) \> r(x) \dd x 
    }\nonumber\\*\quad%
    & =  & \frac{1}{\beta} \int \bigl( f(x) \wedge \mathsf{M} \bigr)
    \> r(x) \dd x
    \\
    & = &  \frac{1}{\beta}  \int \Bigl( f(x) - \bigl( f(x) - f(x)
    \wedge \mathsf{M} \bigr) \Bigr) \> r(x) \dd x
    \\
    & = &  \frac{1}{\beta}  \int f(x) \> r(x) \dd x \nonumber\\
    && + \frac{1}{\beta} \int \Bigl(
    f(x) -  \bigl(f(x) \wedge \mathsf{M} \bigr) \Bigr)
    \>\bigl(-r(x)\bigr) \dd x  
    \\
    & \le &  \frac{1}{\beta} \Gamma
    + \frac{1}{\beta}\int \Bigl( f(x) -  \bigl( f(x) \wedge \mathsf{M} \bigr)
    \Bigr) \> |r(x)| \dd x 
    \\
    & \le & \frac{1}{\beta} \> \Gamma + \frac{1}{\beta}\eps \\
    & \leq & \Gamma + \frac{\eps}{1-\eps} |\Gamma| + \frac{\eps}{1 - \eps},
    \IEEEyesnumber
    \label{eq:cost_tilde}    
  \end{IEEEeqnarray*}
  where the first inequality follows from~\eqref{eq:cons13}; the
  second from~\eqref{eq:15b}; and the last
  from~\eqref{eq:bound_onBeta}.

  We next study $h(\tilde{f})$. Starting with the definition of
  $\tilde{f}$,
  \begin{IEEEeqnarray*}{rCl}
    h(\tilde{f}) & = & 
    \int \frac{1}{\beta} \bigl(f(x)\wedge \mathsf{M} \bigr) \log
    \frac{\beta}{f(x)\wedge \mathsf{M}} \dd x
    \\
    & = & \log\beta + \frac{1}{\beta} \int \bigl(f(x)\wedge \mathsf{M}
    \bigr) \log \frac{1}{f(x)\wedge \mathsf{M}} \dd x
    \\
    & = & \log\beta 
    + \frac{1}{\beta} \int_{x\colon f(x) \le 1} \bigl(f(x)\wedge
    \mathsf{M} \bigr) \log \frac{1}{f(x)\wedge \mathsf{M}} \dd x
    \nonumber\\
    && +\> \frac{1}{\beta} \int_{x\colon f(x) > 1} \bigl(f(x)\wedge
    \mathsf{M} \bigr) \log \frac{1}{f(x)\wedge \mathsf{M}} \dd x.
    \IEEEyesnumber
    \label{eq:ABC}    
  \end{IEEEeqnarray*}
  By~\eqref{eq:defM}, $f(x) \wedge \mathsf{M} = f(x)$ whenever $f(x)
  \le 1$, so
  \begin{IEEEeqnarray}{rCl}
    \IEEEeqnarraymulticol{3}{l}{%
      \int_{x\colon f(x) \le 1} \bigl(f(x)\wedge
    \mathsf{M} \bigr) \log \frac{1}{f(x)\wedge \mathsf{M}} \dd x
    }\nonumber\\*\quad%
    & = & \int_{x\colon f(x) \le 1} f(x) \log \frac{1}{f(x)} \dd x.
    \label{eq:AA}
  \end{IEEEeqnarray}
  Since $\xi\log\xi^{-1}$ is decreasing for $\xi > 1$, and since $f(x)
  >1$ implies $f(x) \wedge \mathsf{M} > 1$ (by \eqref{eq:15c}), 
  \begin{IEEEeqnarray*}{rCl}
    \bigl(f(x)\wedge
    \mathsf{M} \bigr) \log \frac{1}{f(x)\wedge \mathsf{M}}
    \ge f(x) \log \frac{1}{f(x)}, \quad \Bigl( f(x) > 1 \Bigr)
  \end{IEEEeqnarray*}
  and hence
  \begin{IEEEeqnarray}{rCl}
    \IEEEeqnarraymulticol{3}{l}{%
      \int_{x\colon f(x) > 1} \bigl(f(x)\wedge
    \mathsf{M} \bigr) \log \frac{1}{f(x)\wedge \mathsf{M}} \dd x
    }\nonumber\\*\quad%
    & \ge & \int_{x\colon f(x) > 1} f(x) \log \frac{1}{f(x)} \dd x.
    \label{eq:AB}
  \end{IEEEeqnarray}
  Summing \eqref{eq:AA} and \eqref{eq:AB} we obtain
  \begin{IEEEeqnarray}{rCl}
    \int \bigl(f(x)\wedge \mathsf{M}
    \bigr) \log \frac{1}{f(x)\wedge \mathsf{M}} \dd x
    & \ge & h(f).
  \end{IEEEeqnarray}
  
  Using this, \eqref{eq:ABC}, and~\eqref{eq:bound_onBeta} we conclude
  that
  \begin{equation*}
    h(\tilde{f}) = h(f), \quad \text{whenever $h(f) = \infty$}
  \end{equation*}
  and
  \begin{multline}
    h(\tilde{f}) \ge  \log(1-\eps) + h(f) - \frac{\eps}{1-\eps} 
    |h(f)|, 
    \\ \quad \text{whenever $|h(f)| < \infty$}.
    \label{eq:Go2}
  \end{multline}
  And obviously $h(\tilde{f}) \geq h(f)$ whenever $h(f) = -\infty$.
  
  The result now follows by choosing $\eps$ small enough to guarantee
  that the RHS of \eqref{eq:cost_tilde} does not exceed $\Gamma +
  \delta$ and---if $h(f)$ is finite---that the RHS of \eqref{eq:Go2}
  exceeds $h(f)-\delta$.
\end{proof}

The following lemma addresses the case
where~\eqref{eq:cost_integrable} does not hold.
\begin{lemma}
  \label{lem:onemore}
  Let the density $f$ supported by $\set{S}$ be such that
  \begin{equation}
    \label{eq:cost_minus_infty}
    \int f(x) \> r(x) \dd x = -\infty
  \end{equation}
  and $h(f)$ is defined and exceeds $-\infty$
  \begin{equation}
    h(f) > -\infty.
  \end{equation}
  Then there exists a sequence of densities $\{\tilde{f}_k\}$
  supported by $\set{S}$ for which
  \begin{IEEEeqnarray*}{rCl}
    \int \tilde{f}_k(x) \> |r(x)| \dd x & < & \infty,
    \\
    \lim_{k\to\infty} h(\tilde{f}_k) & = & h(f),
    \\
    \noalign{\noindent and \vspace{\jot}}
    \lim_{k\to\infty} \int \tilde{f}_k(x) \> r(x) \dd x & = & -\infty.
  \end{IEEEeqnarray*}
\end{lemma}
\begin{proof}
  Define $r^+ \eqdef \max\{r,0\}$ and $r^- \eqdef \max\{-r,0\}$, so $r
  = r^+ - r^-$ with $r^+(x), r^-(x) \ge 0$. By~\eqref{eq:cost_minus_infty},
  \begin{subequations}
  \begin{equation}
    \int f(x) \> r^-(x) \dd x =\infty,
    \label{eq:f3}
  \end{equation}
  \begin{equation}
    \int f(x) \> r^+(x) \dd x < \infty.
    \label{eq:f378}
  \end{equation}
\end{subequations}
    
  Define for every $k\in\Naturals$
  \begin{equation}
    \label{eq:def_set_Dk}
    \set{D}_k \eqdef \bigl\{ x \colon r^-(x) \le k \bigr\}.
  \end{equation}
  By the MCT
  \begin{IEEEeqnarray}{rCl}
    \lim_{k\to\infty} \int_{\set{D}_{k}}  f(x) \> r^+(x) \dd x
    & = & \int f(x) \> r^+(x) \dd x
    \nonumber
    \\
    & < & \infty
    \IEEEyesnumber
    \label{eq:f2}
    \IEEEyessubnumber
    \label{eq:f2a}
  \end{IEEEeqnarray}
  and 
  \begin{IEEEeqnarray}{rCl}
    \lim_{k\to\infty} \int f(x) \> r^-(x) \I{x\in\set{D}_k} \dd x
    & = & \infty.
    \IEEEyessubnumber
    \label{eq:f2b}
  \end{IEEEeqnarray}
  Consequently,
  \begin{IEEEeqnarray}{rCl}
    \lim_{k\to\infty} \int_{\set{D}_k} f(x) \> r(x) \dd x
    & = & -\infty.
    \label{eq:f4}
  \end{IEEEeqnarray}
  
  The lemma's hypotheses guarantee that $h(f)$ is defined and exceeds
  $-\infty$. Consequently,
  \begin{equation*}
    h(f) = h^+(f) - h^-(f),
  \end{equation*}
  with
  \begin{equation}
    h^-(f) < \infty, \quad 
    h^+(f) \le \infty,
  \end{equation}
  where,
  \begin{IEEEeqnarray*}{rCl}
    h^+(f) & \eqdef & \int f(x) \log \frac{1}{f(x)} \I{f(x) \le 1} 
    \dd x,
    \\
    h^-(f) & \eqdef & \int f(x) \log f(x) \I{f(x) > 1}
    \dd x.
  \end{IEEEeqnarray*}
  By the MCT
  \begin{IEEEeqnarray*}{rCl}
    \int_{\set{D}_k} f(x) \log\frac{1}{f(x)} \I{f(x) \le 1} \dd x 
    & \uparrow & h^+(f)
    \\
    \noalign{\noindent and \vspace{\jot}}
    \int_{\set{D}_k} f(x) \log f(x) \I{f(x) > 1} \dd x 
    & \uparrow & h^-(f)    
  \end{IEEEeqnarray*}
  so, upon subtracting (and recalling $h^-(f) < \infty$)
  \begin{IEEEeqnarray}{rCl}
    \lim_{k\to\infty} \int_{\set{D}_k} f(x) \log \frac{1}{f(x)} \dd
    x & = & h(f).
    \label{eq:f10}
  \end{IEEEeqnarray}

  Define 
  \begin{equation*}
    \beta_k \eqdef \int_{\set{D}_k} f(x) \dd x.
  \end{equation*}
  Note that since $f$ is a density,
  \begin{equation*}
    \beta_k \le 1
  \end{equation*}
  and (by the MCT)
  \begin{equation}
    \beta_k \uparrow 1.
    \label{eq:f5}
  \end{equation}
  Consequently, 
  \begin{equation}
    \label{eq:range_beta_k}
    0 < \beta_k \le 1, \quad \text{$k$ large.}
  \end{equation}
  For every such sufficiently large $k$, define the density
  \begin{equation*}
    \tilde{f}_k(x) \eqdef \beta_{k}^{-1} f(x) \I{x\in\set{D}_k}.
  \end{equation*}
  It is supported by $\set{S}$, and its entropy $h(\tilde{f}_k)$ can
  be expressed as
  \begin{IEEEeqnarray*}{rCl}
    h(\tilde{f}_k) & = & \int \tilde{f}_k(x) \log \frac{1}{\tilde{f}_k(x)} \dd x\\
    & = & \int_{\set{D}_k} \tilde{f}_k(x) \log
    \frac{1}{\tilde{f}_k(x)} \dd x 
    \\
    & = & \int_{\set{D}_k} \frac{1}{\beta_k} \> f(x) \log
    \frac{\beta_k}{f(x)} \dd x 
    \\
    & = & \log\beta_k + \frac{1}{\beta_k} \int_{\set{D}_k} f(x) \log
    \frac{1}{f(x)} \dd x.
  \end{IEEEeqnarray*}
  From this, \eqref{eq:f10}, and \eqref{eq:f5} we obtain
  \begin{IEEEeqnarray}{rCl}
    \lim_{k\to\infty} h(\tilde{f}_k) & = & h(f).
    \label{eq:f15}
  \end{IEEEeqnarray}
  And as to the expectation of $r(x)$ under $\tilde{f}_k$:
  \begin{IEEEeqnarray*}{rCl}
    \IEEEeqnarraymulticol{3}{l}{%
      \int \tilde{f}_k(x) \> r(x) \dd x 
    }\nonumber\\*\quad%
    & = & \frac{1}{\beta_k}\int_{\set{D}_k} f(x) \> r(x) \dd x
    \\
    & = & \frac{1}{\beta_k}\int_{\set{D}_k} f(x) \> r^+(x) \dd x
    - \frac{1}{\beta_k}\int_{\set{D}_k} f(x) \> r^-(x) \dd x.
  \end{IEEEeqnarray*}
  The first term on the LHS is finite by~\eqref{eq:range_beta_k} and
  \eqref{eq:f378}. The second tends to $-\infty$ by~\eqref{eq:f5} and
  \eqref{eq:f4}. Hence,
  \begin{IEEEeqnarray}{rCl}
    \lim_{k\to\infty} \int \tilde{f}_k(x) \> r(x) \dd x 
    & = & -\infty.
    \label{eq:f20}
  \end{IEEEeqnarray}
  Moreover,
  \begin{IEEEeqnarray*}{rCl}
    \IEEEeqnarraymulticol{3}{l}{%
      \int \tilde{f}_k(x) \> |r(x)| \dd x 
    }\nonumber\\*\quad%
    & = & \frac{1}{\beta_k} \int_{\set{D}_k} f(x) \> r^+(x) \dd x
    + \frac{1}{\beta_k} \int_{\set{D}_k} f(x) \> r^-(x) \dd x
    \\
    & \le & \frac{1}{\beta_k} \int f(x) \> r^+(x) \dd x + k 
    \\
    & < & \infty,
    \IEEEyesnumber
    \label{eq:f25}
  \end{IEEEeqnarray*}
  where the first inequality follows from the nonnegativity of $r^{+}$
  and from the definition of the set
  $\set{D}_{k}$~\eqref{eq:def_set_Dk}, and the second inequality
  follows from~\eqref{eq:f378} and \eqref{eq:range_beta_k}.
  
  The lemma now follows from \eqref{eq:f25}, \eqref{eq:f15}, and \eqref{eq:f20}.
\end{proof}

\begin{proof}[Proof of Proposition~\ref{prop:CanAssume}]
  Since $\Gamma$ exceeds $\Gamma_{0}$, it follows
  from~\eqref{block:amos} that
  \begin{equation}
    \label{eq:h_is_finite}
    -\infty < h^{\star}(\Gamma) < \infty.
  \end{equation}
  Let the density $f$ nearly achieve $h^{\star}(\Gamma)$ in the
  sense that it is supported by $\set{S}$ and that
  \begin{equation}
    \label{eq:f_almost_achieve17}
    \int f(x) \, r(x) \dd{x} \leq \Gamma, \quad\text{and} \quad h(f) >
    h^{\star}(\Gamma) - \frac{\delta}{2}.
  \end{equation}
  By~\eqref{eq:h_is_finite}, \eqref{eq:f_almost_achieve17}, and the
  definition of $h^{\star}(\Gamma)$,
  \begin{equation}
    \label{eq:h_of_f_is_finite}
    -\infty < h(f) < \infty.
  \end{equation}

  If $\int f(x) \abs{r(x)} \dd{x}$ is finite, then the result follows
  directly from Lemma~\ref{lem:anotherlemma}. It remains to prove the
  result when this integral is infinite.
  In this case $\int f(x) \> r(x) \dd x = -\infty$
  by~\eqref{eq:f_almost_achieve17} (because $\Gamma < \infty$). Using
  this, the finiteness of $h(f)$ \eqref{eq:h_of_f_is_finite}, and
  Lemma~\ref{lem:onemore}, we infer the existence of a density
  $\tilde{f}$ that supported by $\set{S}$ and for which
  \begin{subequations}
    \label{block:ball}
    \begin{equation}
      \label{eq:ball109}
      \int \tilde{f}(x) \> |r(x)| \dd x  <  \infty,
    \end{equation}
    \begin{equation}
      \label{eq:ball100}
      h(\tilde{f})  >   h(f) - \frac{\delta}{2},
    \end{equation}
    \begin{equation}
      \label{eq:ball110}
      \int \tilde{f}(x) \> r(x) \dd x  <  \Gamma.
    \end{equation}
  \end{subequations}
  Applying Lemma~\ref{lem:anotherlemma} to the density $\tilde{f}$, we
  conclude that there exists a bounded density $f^{\star}$ that is supported
  by $\set{S}$ and that satisfies
  \begin{equation}
    h(f^{\star}) > h(\tilde{f}) -\frac{\delta}{2} \quad \text{and}
    \quad \int f^{\star}(x) \> r(x) \dd{x} \leq \Gamma + \delta
  \end{equation}
  and hence, in view of~\eqref{block:ball} and~\eqref{eq:f_almost_achieve17}, 
\begin{equation}
    h(f^{\star}) > h^{\star}(\Gamma) - \delta \quad \text{and}
    \quad \int f^{\star}(x) \> r(x) \dd{x} \leq \Gamma + \delta.
  \end{equation}
  The existence of $f^{\star}$ concludes the proof of the proposition
  for the case where $\int f(x) \> \abs{r(x)} \dd{x}$ is infinite.
\end{proof}

\section*{Acknowledgment}
Discussions with Stefan M.\ Moser and Igal Sason are gratefully acknowledged.



\end{document}